\documentclass[12pt]{article}
\usepackage{amsthm,amsfonts,amsmath, amscd, bbold}
\usepackage{accents}
\usepackage{authblk}

                                \usepackage{verbatim}
\swapnumbers
                              
\theoremstyle{plain}

\usepackage{enumitem}
\numberwithin{equation}{section}
\newtheorem{theorem}{Theorem}[section]
\newtheorem{conjecture}[theorem]{Conjecture}

\newtheorem{lemma}[theorem]{Lemma}

\newtheorem{corollary}[theorem]{Corollary}

\theoremstyle{definition}
\newtheorem{definition}[theorem]{Definition}
\newtheorem{example}[theorem]{Example}

\theoremstyle{remark}
\newtheorem{remark}[theorem]{Remark}
\numberwithin{equation}{section}

\newcommand{\R}{{\mathbb R}}
\newcommand{\bR}{{\mathbb R}}

\newcommand{\cB}{{\mathcal B}}

\newcommand{\cH}{{\mathcal H}}

\newcommand{\ket}[1]{\left\vert #1\right\rangle}
\newcommand{\bra}[1]{\left\langle #1\right\vert}

\renewcommand{\Re}{\,\mathrm{Re}\,}   
\renewcommand{\Im}{\,\mathrm{Im}\,}   
\def\idty{{\mathchoice {\mathrm{1\mskip-4mu l}} {\mathrm{1\mskip-4mu l}} %
{\mathrm{1\mskip-4.5mu l}} {\mathrm{1\mskip-5mu l}}}}

\newcommand{\Tr}{\mathrm{Tr}}

\newcommand{\be}{\begin{equation}}
\newcommand{\ee}{\end{equation}}
\newcommand{\bea}{\begin{eqnarray}}
\newcommand{\eea}{\end{eqnarray}}
\newcommand{\beann}{\begin{eqnarray*}}
\newcommand{\eeann}{\end{eqnarray*}}

\usepackage{color}

\begin{document}

\title{Upper continuity bound on the quantum quasi-relative entropy}
\author{Anna Vershynina}
\affil{\small{Department of Mathematics, Philip Guthrie Hoffman Hall, University of Houston, 
3551 Cullen Blvd., Houston, TX 77204-3008, USA}}
\renewcommand\Authands{ and }
\renewcommand\Affilfont{\itshape\small}

\date{\today}

\maketitle

\begin{abstract}
We provide an upper bound on the quasi-relative entropy in terms of the trace distance. The bound is derived for two cases: 1) any operator monotone decreasing function and full rank mixed qubit or classical states; 2) a large class of operator monotone decreasing function and any mixed qubit or classical states. Moreover, we derive an upper bound for the Umegaki and Tsallis relative entropies in the case of any finite-dimensional states. The bound for the relative entropy improves the known bounds for some states in any dimensions larger than four. The bound for the Tsallis entropy improves the known bounds.
\end{abstract}

\section{Introduction}

Quantum quasi-relative entropy was introduced by Petz \cite{P85, P86} as a quantum generalization of a classical Csisz\'ar's $f$-divergence \cite{C67-2}. It is defined in the context of von Neumann algebras, but we  consider only the finite-dimensional Hilbert space setup. Let $\cH$ be a  finite-dimensional  Hilbert space, $\rho$ and $\sigma$ be two states (given by density operators), and $f:(0,\infty)\rightarrow\mathbf{R}$ be an operator convex function. Then the quasi-relative entropy (often times called $f$-divergence) is defined as
\begin{equation}\label{eq:quasi-1}
S_f(\rho||\sigma)=\Tr(f(\Delta_{\sigma,\rho})\rho)\ ,
\end{equation}
where $\Delta_{\sigma,\rho}$ is a relative modular operator defined by Araki \cite{A76} that acts as a left and right multiplication for the positive invertible operators $A$ and $B$
$$\Delta_{A,B}(X)=L_AR_B^{-1}(X)=AXB^{-1}\ .$$
Throughout the paper we consider $\rho$ and $\sigma$ to be strictly positive density operators. 

Note that there is an equivalent definition of the $f$-divergence that is sometimes used
\begin{equation}\label{eq:quasi-2}
\tilde{S}_f(\rho||\sigma)=\Tr(f(\Delta_{\rho,\sigma})\sigma)=S_f(\sigma\|\rho)\ .
\end{equation}
But taking $g(x)=xf(x^{-1})$, we obtain
$$S_g(\rho\|\sigma)=\tilde{S}_f(\rho\|\sigma)\ . $$
Since all references related to quasi-relative entropies cited in this paper use definition (\ref{eq:quasi-1}), we will also use this one.

Taking the logarithmic function $f(x)=-\log(x)$ reduces quasi-relative entropy to the Umegaki relative entropy \cite{U62},
$$S(\rho\|\sigma)=\Tr(\rho[\log\rho-\log\sigma])\ . $$

A famous bound relating the quantum relative entropy and the trace distance between two quantum states, is called the Pinsker inequality. A similar inequality holds for the quasi-relative entropy as well, as was shown by Hiai and Mosonyi \cite{HM16} :
$$\frac{f''(1)}{2}\|\rho-\sigma \|_1^2\leq S_f(\rho\|\sigma)\ . $$

The questions now, is to obtain the upper bound on the quasi-relative entropy in terms of the trace distance. The upper continuity bound for the Umegaki relative entropy was obtained in \cite{AE11} in the following form:
\begin{equation}\label{eq:rel_known}
S(\rho\|\sigma)\leq (\alpha_\sigma+T)\log(1+T/\alpha_\sigma)-\alpha_\rho\log(1+T/\alpha_\rho)\ ,
\end{equation}
where throughout the paper $\alpha_\omega$ is the minimal non-zero eigenvalue of the state $\omega$, and $T:=\|\rho-\sigma\|_1/2$.

For $0<q\neq1$, taking $f(x)=\frac{1}{1-q}(1-x^{1-q})$ in (\ref{eq:quasi-1}) leads  to the Tsallis $q$-entropy defined as
\begin{equation}\label{eq:q}
S_q(\rho\|\sigma)=\frac{1}{1-q}\left(1-\Tr(\rho^q\sigma^{1-q})\right)\ ,
\end{equation}
for $\ker(\sigma)\subset\ker(\rho)$. A series of upper bounds for the $q$-entropy in terms of the trace distance were obtained in \cite{R10}.  The derived bounds in \cite{R10} are, in particular, the following:
\begin{itemize}
\item for $q>1$
$$S_q(\rho\|\sigma)\leq \frac{\lceil q \rceil-1}{q-1}\frac{\lambda^{q-1}}{\alpha_\sigma^{q-1}}\|\rho-\sigma\|_1\ , $$
where $\lambda$ is the maximum eigenvalue in the joint spectra of $\rho$ and $\sigma$, and $\lceil q \rceil$ is the smallest integer that is larger than $q$;
\item  for  $1<q\leq 2$, and denoting $\lambda_\rho$ to be the maximal eigenvalue of $\rho$,  and $\alpha=\min\{\alpha_\rho,\alpha_\sigma\}$, the following bounds hold
\begin{equation}\label{eq:Rmoreold}
S_q(\rho\|\sigma)\leq \frac{1}{q-1}\frac{\lambda_\rho^{q}}{\alpha^{q}}\, \|\rho-\sigma\|_1\ . 
\end{equation}
\item for $0<q<1$, 
\begin{equation}\label{eq:Rless}
S_q(\rho\|\sigma)\leq \frac{1}{{1-q}}\, \frac{\lambda_\rho^{q}}{\alpha_\sigma^{q}}\, \|\rho-\sigma\|_1\ . 
\end{equation}
\end{itemize}
Note that a series of other bounds was derived in \cite{R10}, which for some states could be  an improvement of the bounds above.

We investigate the upper continuity bound for a quasi-relative entropy for an operator monotone decreasing function $f$. 

Main results:
\begin{itemize}
\item Let $f$ be an operator monotone decreasing function, and states $\rho$ and $\sigma$ are either $2$-dimensional qubit states or classical states. Assume one of the two conditions: 1) $\rho$ is full rank; 2) $f$ is such that $a_f=0$ (defined below). in Theorem \ref{thm:main3} we prove
\begin{equation}\label{eq:prelim}
S_f(\rho\|\sigma)\leq\|\rho-\sigma\|_1\left[\frac{\lambda_\rho}{\lambda_\rho-\alpha_\sigma} f(\lambda_\rho^{-1}\alpha_\sigma)-a_f\right] \ ,
\end{equation}
where
\begin{itemize}
\item $\lambda_\rho\in(0,1]$ is the largest eigenvalue of $\rho$,
\item $\alpha_\sigma\in(0,1]$ is the smallest eigenvalue of $\sigma$,
\item $a_f=-\lim_{y\uparrow \infty}\frac{f(iy)}{iy}.$
\end{itemize}
In the most general case, we obtain an upper bound dependent on the dimension of the Hilbert space, see Theorem \ref{thm:main}. We  conjecture that the bound \eqref{eq:prelim} holds in the general case as well, see Conjecture \ref{conjecture}.

\item In Theorem \ref{thm:main2} we provide the following upper bound to the relative entropy
$$S(\rho\|\sigma)\leq\|\rho-\sigma\|_1\lambda_\rho \frac{\log(\alpha_\rho)-\log(\alpha_\sigma)}{\alpha_\rho-\alpha_\sigma}\leq \frac{\lambda_\rho}{\alpha} \|\rho-\sigma\|_1\ ,
$$
where $\alpha=\min\{\alpha_\rho,\alpha_\sigma\}$. For any dimension larger than four, there are states, for which the present bound is  better than known bound (\ref{eq:rel_known}). See Section \ref{rem:rel-better}.

If states are two-dimensional, from (\ref{eq:prelim}) and Corollary \ref{cor:rel}, we obtain
$$S(\rho\|\sigma)\leq \|\rho-\sigma\|_1 \lambda_\rho \frac{\log\lambda_\rho-\log\alpha_\sigma}{\lambda_\rho-\alpha_\sigma}\leq\frac{\lambda_\rho}{\alpha_\sigma} \|\rho-\sigma\|_1\ .
$$
\item For Tsallis relative entropy for $q>1$, in Remark \ref{rem} we show how improve the bound (\ref{eq:Rmoreold}) in the proof in \cite{R10}. We obtain
\begin{equation}\label{eq:Rmore}
S_q(\rho\|\sigma)\leq \frac{\lambda_\rho^{q}}{\alpha^{q}}\, \|\rho-\sigma\|_1\ . 
\end{equation}

\item In Theorem \ref{thm:main1}, for $q\in(0,1)$ we obtain the upper bound on the $q$-entropy
$$S_q(\rho\|\sigma)\leq\frac{1}{1-q} \|\rho-\sigma\|_1\lambda_\rho^q \frac{\alpha_\rho^{1-q}-\alpha_\sigma^{1-q}}{\alpha_\rho-\alpha_\sigma}\leq\|\rho-\sigma\|_1\frac{\lambda_\rho^q}{\alpha^q}\ . $$
This bound is clearly an improvement of (\ref{eq:Rless}), since it improves the constant.

If the states $\rho$ and $\sigma$ are two-dimensional, then from Theorem \ref{thm:main} (see Section \ref{sec:Ts-qubits}), we obtain
$$ S_{q}(\rho\|\sigma)\leq \frac{1}{1-q}\,\|\rho-\sigma\|_1\lambda_\rho^{q} \, \frac{\lambda_\rho^{1-q}-\alpha_\sigma^{1-q}}{\lambda_\rho-\alpha_\sigma}\leq \|\rho-\sigma\|_1\, \frac{\lambda_\rho^{q}}{\alpha_\sigma^{q}} \ . $$
\end{itemize}

\section{Preliminaries}
\subsection{Operator monotone functions}
\begin{definition}\label{def:op-mon}
A function $f: (a,b) \to \R$ is {\em operator monotone} if for any pair of  self-adjoint operators 
 $A$ and $B$ on some 
 Hilbert space that have spectrum in $(a,b)$, the operator $$f(A) -f(B)\geq0$$ is positive semidefinite whenever $A - B\geq 0$ is positive semidefinite. We say that $f$ is {\it operator monotone decreasing} on $(a,b)$ in case $-f$ is 
 operator monotone.
\end{definition}

\begin{definition}\label{def:convex}
A function  $f: (a,b) \to \R$ is {\em operator concave} on the positive operators, when  for all positive semidefinite 
operators 
 $A$ and $B$ on some 
 Hilbert space that have spectrum in $(a,b)$ and all $\lambda$ in $(0,1)$, $$ f((1-\lambda)A + 
\lambda B))-(1-\lambda)f(A) - \lambda f(B)\geq 0$$
is positive semidefinite. A function $f$ is {\em operator convex} if $-f$ is operator concave.
\end{definition}

\begin{theorem}[Bhatia '97 ]\cite[Theorem V.2.5]{B97}\label{thm:Bhatia}
Every operator monotone function $f:[0,\infty)\rightarrow \mathbb{R}$ is operator concave. Moreover, every continuous function $f$ mapping $[0,\infty)\rightarrow [0,\infty)$ into itself is operator monotone if and only if it is operator concave. 
\end{theorem}

\begin{example} 
Note that
\begin{itemize} 
\item $f(x)=\log x$ is operator monotone;
\item $f(x)=x\log x$ is operator convex.
\end{itemize}
\end{example}

\begin{example}\label{ex-power}
Let $f(x)=x^p$, where ${p}\in\bR$. Then by \cite[Theorem V.2.10]{B97} the function $f$ is 
\begin{enumerate}
\item operator monotone and operator concave if and only if $p\in[0,1]$;
\item operator convex if and only if $p\in[-1,0]\cup[1,2]$;
\item operator monotone decreasing and operator convex  if and only if $p\in[-1,0]$.
\end{enumerate}
\end{example}

\begin{definition}\label{def:pick}
 A {\em Pick function} is a function $f$ that is analytic on the upper half plane and has a positive imaginary part. The set of Pick functions on $(a,b)$ is denoted as $\mathcal{P}_{(a,b)}.$
\end{definition}

\begin{theorem}[L\"{o}wner '34] \cite[Theorem V.4.7]{B97} A function $f$ on $(a,b)$ is operator monotone if and only if $f$ is a restriction of a Pick function $f\in\mathcal{P}_{(a,b)}$ to $(a,b)$.
\end{theorem}

\begin{corollary}
A function $f$ on $(0,\infty)$ is operator monotone decreasing  if and only if $-f\in\mathcal{P}_{(0,\infty)}$.
\end{corollary}
Denote the set of operator monotone decreasing functions $f$ (i.e.$-f\in\mathcal{P}_{(0,\infty)}$) as $\mathcal{Q}_{(0,\infty)}$.

\begin{example}
From \cite[Exercise V.4.8]{B97} The following functions belong to $\mathcal{Q}_{(0,\infty)}$: 
\begin{itemize}
\item $f(x)=-\log x$, 
\item $f(x)=-x^p$ for $p\in[0,1]$, 
\item $f(x)=x^p$ for $p\in[-1,0]$.
\end{itemize}
\end{example}

According to \cite[Chapter II, Theorem I]{Dono} every function  $f\in\mathcal{Q}_{(0,\infty)}$, has a canonical integral representation 
\begin{equation}\label{low}
f(x) = -a x- b +\int_{0}^\infty \left(  \frac{1}{t +x }-\frac{t}{t^2+1}   \right){\rm d}\mu_f(t)\ ,
\end{equation}
where  $a:=-\lim_{y\uparrow\infty}\frac{f(iy)}{iy}\geq 0$, $b:=-\Re f(i)\in\bR$ and $\mu$ is a positive measure on $(0,\infty)$ such that 
${\displaystyle \int_{0}^\infty  \frac{1}{t^2+1}{\rm d}\mu_f(t)<\infty}$, and
 \begin{equation}\label{muform}
 \mu_f(x_1) - \mu_f(x_0) = -\lim_{y\downarrow 0} \frac{1}{\pi} \int_{x_0}^{x_1} \Im f(-x+iy){\rm d} x\ .
 \end{equation}
Conversely, every such function belongs in $\mathcal{Q}_{(0,\infty)}$.

We consider functions $f\in\mathcal{Q}_{(0,\infty)}$ such that $f(1)=0$. The last condition is equivalent to
$$0=f(1) = -a - b +\int_{0}^\infty \left(  \frac{1}{t +1 }-\frac{t}{t^2+1}   \right){\rm d}\mu_f(t)\ ,
 $$
 in other words,
 \begin{equation}\label{eq:f10}
 a+b=\int_{0}^\infty \left(  \frac{1}{t +1 }-\frac{t}{t^2+1}   \right){\rm d}\mu_f(t)\ .
 \end{equation}
 Therefore, the operator monotone decreasing function $f$ such that $f(1)=0$ has the following integral representation
 \begin{equation}\label{low2}
f(x) = a(1-x)+ \int_{0}^\infty \left(  \frac{1}{t +x }- \frac{1}{t +1 }\right){\rm d}\mu_f(t)\ .
\end{equation}

\begin{example}\label{ex-power}
Consider the power function $f(x)=-x^p$ for $p\in(0,1)$. It is operator monotone decreasing. Then 
$$a = -\lim_{y\uparrow\infty}f(iy)/(iy) = 0\ , \ \ \text{and } \ b = {\cos}(p \pi/2)\ .$$
For $x>0$, $\lim_{y\downarrow 0} \Im f(-x + iy) = -x^p\sin(p \pi)$ so that
$${\rm d}\mu(x) = \pi^{-1}\sin(p \pi) x^p{\rm d}x\ .$$ This yields the representation
\begin{equation}\label{eq:powex}
-x^p =  -{\cos}(p \pi/2)  + \frac{\sin(p \pi)}{\pi} \int_{0}^\infty t^p\left(  
\frac{1}{t +x } -\frac{t}{t^2+1}  \right){\rm d}t \ .
\end{equation}

\begin{example}\label{ex-log} Let $f(x)=-\log(x)$. It is operator monotone decreasing.
Then 
$$b=\Re(\log(i)) = 0\ ,$$ and 
$$a=\lim_{y\uparrow\infty}\log(iy)/(iy) = \lim_{y\uparrow\infty}(\log y + i\pi/2)/(iy) =0\ . $$ It is clear from \eqref{muform} that
$${\rm d}\mu(x) = \frac{1}{\pi}\lim_{y\downarrow 0}\Im\log(-x + iy){\rm d}x = {\rm d}x\ .$$
Then the  integral representation (\ref{low}) gives the following formula for the logarithmic function
\begin{equation}\label{eq:logex}
-\log x=\int_{0}^\infty \left( \frac{1}{t +x } - \frac{t}{t^2+1}  \right){\rm d}t \ ,
\end{equation}
which is also obvious from the direct computation of the integral.
\end{example}
\end{example}

\subsection{Quasi-relative entropy}
\begin{definition}\label{def:qre}
For an operator convex function $f$, such that $f(1)=0$, and strictly positive states $\rho$ and $\sigma$ acting on a finite-dimensional Hilbert space $\cH$, {\it the quasi-relative entropy} (or sometimes referred to as {\it the $f$-divergence}) is defined as 
$$S_f(\rho|| \sigma)=\Tr(f(\Delta_{\sigma,\rho}){\rho})\ ,$$
where the relative modular operator, introduced by Araki \cite{A76}, $$\Delta_{A,B}(X)=L_AR_B^{-1}(X)=AXB^{-1}$$ is a product of left and right multiplication operators, $L_A(X)=AX$ and $R_B(X)=XB$. Throughout this paper we consider finite-dimensional setup, so the operators are invertible. (In general, $A^{-1}$ is stands for the generalized inverse of $A$.)
\end{definition}

There right and left multiplication operators have the following properties \cite{JR10}
\begin{enumerate}
\item They commute, i.e. $$[L_A, R_B]=0, $$
since $$L_AR_B(X)=AXB=R_BL_A(X)\ . $$
\item The operators $L_A$ and $R_A$ are invertible if and only if $A$ is non-singular, giving $L_A^{-1}=L_{A^{-1}}$ and $R_A^{-1}=R_{A^{-1}}$. 
\item If $A$ is self-adjoint, then $L_A$ and $R_A$ are both self-adjoint with respect to the Hilbert Schmidt inner product. 
\item If $A\geq 0$, then  $L_A$ and $R_A$ are positive semi-definite, i.e.
$$\Tr X^*L_A(X)=\Tr X^*AX\geq 0 $$
and 
$$\Tr X^*R_A(X)=\Tr X^*XA=\Tr X^*AX\geq 0\ . $$
\item If $A>0$, for any function $f(0,\infty)\rightarrow \mathbb{R}$, we have $f(L_A)=L_{f(A)}$ and $f(R_A)=R_{f(A)}$. This follows from the spectral decomposition of $A$, denoted as $A=\sum_{j=1}^d \lambda_j \ket{j}\bra{j}$. Then for any $j,k=1,\dots, d$ the operator $\ket{j}\bra{k}$ is the eigenstate of the operator $L_A$ (and $R_A$) with eigenvalue $\lambda_j$ (or $\lambda_k$). The later has degeneracy $d$
$$L_A\ket{j}\bra{k}=\lambda_j\ket{j}\bra{k},\ \ R_A\ket{j}\bra{k}=\lambda_k\ket{j}\bra{k}\ .$$
Therefore,
$$f(L_A)\ket{j}\bra{k}=f(\lambda_j)\ket{j}\bra{k},\ \ f(R_A)\ket{j}\bra{k}=f(\lambda_k)\ket{j}\bra{k}\ . $$
\end{enumerate}

There is a straightforward way to calculate the quasi-relative entropy from the spectral decomposition of states. Let $\rho$ and $\sigma$ have the following spectral decomposition
\begin{equation}\label{eq:spectral}
\rho=\sum_j\lambda_j\ket{\psi_j}\bra{\psi_j}, \ \ \sigma=\sum_k\mu_k\ket{\phi_k}\bra{\phi_k}\ ,
\end{equation}
where the eigenvalues are ordered: 
$$\lambda_n\leq \dots\leq \lambda_1, \ \ \mu_n\leq \dots\leq \mu_1\ . $$
 the set $\{\ket{\phi_k}\bra{\psi_j}\}_{j,k}$ forms an orthonormal basis of $\cB(\cH)$, the space of bounded linear operators, with respect to the Hilbert-Schmidt inner product defined as $\langle A, B \rangle=\Tr(A^*B)$. By \cite{V16}, the modular operator can be written as
\begin{equation}\label{eq:modular}
\Delta_{\sigma, \rho}=\sum_{j,k} \frac{\mu_k}{\lambda_j}P_{j,k}\ ,
\end{equation}
where $P_{j,k}:\cB(\cH)\rightarrow\cB(\cH)$ is defined by
$$P_{j,k}(X)=\ket{\psi_j}\bra{\phi_k}\bra{\psi_j}X\ket{\phi_k}\ . $$
The quasi-relative entropy is calculated as follows
\begin{equation}\label{eq:formula}
S_f(\rho||\sigma)=\sum_{j,k}\lambda_j f\left(\frac{\mu_k}{\lambda_j}\right)|\bra{\phi_k}\ket{\psi_j}|^2\ . 
\end{equation}

\begin{example} For $f(x)=-\log x$, the quasi-relative entropy becomes the Umegaki relative entropy
$$S_{-\log}(\rho\|\sigma)=S(\rho\|\sigma)=\Tr (\rho\log\rho-\rho\log\sigma)\ . $$
\end{example}

\begin{example}
For $p\in(-1,2)$ and $p\neq 0,1$ let us take the function 
$$f_p(x):=\frac{1}{p(1-p)}(1-x^p)\ ,$$
which is {operator} convex. The quasi-relative entropy for this function is calculated to be
$$S_{f_p}(\rho|| \sigma)=\frac{1}{p(1-p)}\left(1-\Tr(\sigma^{p}\rho^{1-p})\right)\ .$$
\end{example}

\begin{example}
For $p\in(-1,1)$ take $q=1-p\in(0,2)$, the function
$$f_q(x)=\frac{1}{1-q}(1-x^{1-q}) $$
is operator convex. The quasi-relative entropy for this function is known as Tsallis $q$-entropy
$$S_q(\rho\|\sigma)= \frac{1}{1-q}\left(1-\Tr(\rho^{q}\sigma^{1-q})\right)\ .$$
\end{example}

\section{Upper continuity bound for qubits}

Consider a case when $\rho$ and $\sigma$ are $2$-dimensional states, or diagonalizable in the same basis states on $d$-dimensional Hilbert space. Then for any operator monotone decreasing function $f$ the following upper bound holds.

\begin{theorem}\label{thm:main3}
Let $f\in\mathcal{Q}_{(0,\infty)}$ be an operator monotone decreasing function such that $f(1)=0$. Let $\rho$ and $\sigma$ be two strictly positive density operators on a $2$-dimensional Hilbert space (qubits), or states diagonalizable in the same basis on any finite-dimensional Hilbert space (classical states). Assume one of the two conditions: 1) $\rho$ is full rank; 2) $f$ is such that $a_f=0$ (defined below). Then 
\begin{equation}\label{eq:main}
S_f(\rho\|\sigma)\leq\|\rho-\sigma\|_1\left[\frac{\lambda_\rho}{\lambda_\rho-\alpha_\sigma} f(\lambda_\rho^{-1}\alpha_\sigma)-a_f\right] \ ,
\end{equation}
where

\begin{itemize}
\item $\lambda_\rho\in(0,1]$ is the largest eigenvalue of $\rho$,
\item $\alpha_\sigma\in(0,1]$ is the smallest eigenvalue of $\sigma$,
\item $a_f=-\lim_{y\uparrow \infty}\frac{f(iy)}{iy}.$
\end{itemize}
\end{theorem}

\begin{proof}  Every function $f\in\mathcal{Q}_{(0,\infty)}$ (i.e. operator monotone decreasing function), such that $f(1)=0$ admits an integral representation (\ref{low2}). Since $S_f(\rho\|\rho)=0= \Tr(f(\Delta_{\rho,\rho})\rho)$, we have
\begin{align}
S_f(\rho\|\sigma)&=\Tr\{(f(\Delta_{\sigma, \rho})-f(\Delta_{\rho, \rho}))\rho\}\\
&=a_f\Tr\{\Delta_{\rho,\rho}\rho\}-a_f\Tr\{\Delta_{\sigma,\rho}\rho\}+\int_0^\infty d\mu_f(t)\  \Tr\{\left((t\idty+\Delta_{\sigma,\rho})^{-1}-(t\idty+\Delta_{\rho,\rho})^{-1}\right)\rho\}\\
&=\int_0^\infty d\mu_f(t)\  \Tr\{\left((t\idty+\Delta_{\sigma,\rho})^{-1}-(t\idty+\Delta_{\rho,\rho})^{-1}\right)\rho\}\ .
\end{align}
The second equality is due to the fact that either $a_f=0$ or $\rho$ is full rank.

The formula $A^{-1}-B^{-1}=A^{-1}(B-A)B^{-1}$ holds for any invertible operators $A$ and $B$. Using the fact that the modular operator is the product of left and right multiplications, $\Delta_{\sigma,\rho}=L_\sigma R_{\rho^{-1}}$, we obtain
\begin{align}
S_f(\rho\|\sigma)&=\int_0^\infty d\mu_f(t)\  \Tr\{\left((t\idty+\Delta_{\sigma,\rho})^{-1}(L_\rho-L_\sigma)R_{\rho^{-1}}(t\idty+\Delta_{\rho,\rho})^{-1}\right)\rho\}\\
&=\int_0^\infty d\mu_f(t)\  \Tr\{\left((t\idty+\Delta_{\sigma,\rho})^{-1}(L_\rho-L_\sigma)(t\idty+\Delta_{\rho,\rho})^{-1}\right)(I)\}\ .
\end{align}

From (\ref{eq:modular}), the last trace can be written as a trace of a product of two matrices:
$$S_f(\rho\|\sigma)=\int_0^\infty d\mu_f(t)\ (t+1)^{-1}\Tr\{D_t(\rho-\sigma)\}\ ,$$
where, with the spectral decomposition (\ref{eq:spectral}) of $\rho$ and $\sigma$,
$$D_t=\sum_{jk} \left(t+ \frac{\mu_k}{\lambda_j}\right)^{-1}\bra{\psi_j}\ket{\phi_k}\ket{\psi_j}\bra{\phi_k}\ .$$

In Lemma \ref{lemma}  take $X=\rho-\sigma$, $D=D_t$ and $C=\max_{kj}\left(t+ \frac{\mu_k}{\lambda_j}\right)^{-1}=  (t+\lambda_\rho^{-1}\alpha_\sigma)^{-1}$. Then in both cases for states $\rho$ and $\sigma$ specified in Theorem, we obtain
$$
\left|\Tr\{D_t(\rho-\sigma)\} \right| \leq (t+\lambda_\rho^{-1}\alpha_\sigma)^{-1} \|\rho-\sigma\|_1\ .$$
Therefore, in both cases,
$$S_f(\rho\|\sigma)\leq  \|\rho-\sigma\|_1 \int_0^\infty \frac{1}{t+\lambda_\rho^{-1}\alpha_\sigma}\, \cdot\, \frac{1}{t+1}\, d\mu_f(t)\ ,$$
Note that
$$\frac{1}{t+\lambda_\rho^{-1}\alpha_\sigma}\, \cdot\, \frac{1}{t+1}=\frac{\lambda_\rho}{\lambda_\rho-\alpha_\sigma}\left\{\frac{1}{t+\lambda_\rho^{-1}\alpha_\sigma} -\frac{1}{t+1}\right\}\ . $$
Therefore,
\begin{align}
S_f(\rho\|\sigma)&\leq  \|\rho-\sigma\|_1 \frac{\lambda_\rho}{\lambda_\rho-\alpha_\sigma}\int_0^\infty \left\{\frac{1}{t+\lambda_\rho^{-1}\alpha_\sigma}- \frac{1}{t+1}\right\}\, d\mu_f(t)\\
&=\|\rho-\sigma\|_1 \frac{\lambda_\rho}{\lambda_\rho-\alpha_\sigma} \left[f(\lambda_\rho^{-1}\alpha_\sigma)-a_f(1-\lambda_\rho^{-1}\alpha_\sigma)\right] \ ,
\end{align}
where the last equation is obtained from the integral representation (\ref{low2}) of function $f$. Here, recall, $a_f=-\lim_{y\uparrow \infty}\frac{f(iy)}{iy}.$
\end{proof}

\begin{lemma}\label{lemma}
For  orthogonal bases $\{\ket{\psi_j}\}$ and $\{\ket{\phi_k}\}$, let 
\begin{equation}\label{eq:D1}
D=\sum_{kj}C_{kj}\bra{\psi_j}\ket{\phi_k}\ket{\psi_j}\bra{\phi_k}\ ,
\end{equation}
such that $0\leq C_{kj}\leq C$ for all $k,j$ and some $C$.
Consider two cases:
\begin{itemize}
\item Let $X$ be a diagonal matrix in either basis: without loss of generality let $X=\sum_k x_k \ket{\phi_k}\bra{\phi_k}$.
\item Let $X$ be a $2\times 2$ Hermitian traceless matrix, i.e. $X^*=X$ and $\Tr(X)=0$.
\end{itemize}
In both cases, 
$$\left|\Tr(DX)\right|\leq C\|X\|_1\ .$$
\end{lemma}
\begin{proof}
1. Since $X$ is diagonal matrix, the trace norm is the sum of the absolute values of the eigenvalues
$$\|X\|_1=\Tr|X|=\sum_j|x_j|\ . $$
The trace can be calculated 
$$\Tr\{DX\}=\sum_{kj} C_{kj}x_k|\bra{\psi_j}\ket{\phi_k}|^2 . $$
Therefore,
$$\left|\Tr\{DX)\} \right| \leq \sum_{kj}C_{kj}|x_k||\bra{\psi_j}\ket{\phi_k}|^2\leq C\sum_{kj}|x_k||\bra{\psi_j}\ket{\phi_k}|^2\leq C\sum_k|x_k|= C \|X\|_1\ .$$

2. Assume that $X$ is a $2\times2$ Hermitian traceless matrix, such that
$$X=\sum_{ij}x_{ij} \ket{\psi_i}\bra{\psi_j}\ .$$
 First, let us compute the trace norm of $X$. Let $\omega_1$ and $\omega_2$ be the singular values of $X$, then
\begin{align}
\|X\|_1^2&=(\omega_1+ \omega_2)^2\\
&=\Tr(X^*X)+2|\det(X)|\\
&=x_{11}^2+x_{22}^2+|x_{12}|^2+|x_{21}|^2+2\left| x_{11}x_{22}-x_{12}x_{21}\right|\\
&=x_{11}^2+x_{22}^2+|x_{12}|^2+|x_{21}|^2+2(x_{11}^2+|x_{12}|^2)\\
&=2\left(x_{11}^2+x_{22}^2+|x_{12}|^2+|x_{21}|^2 \right)\\
&=2\sum_{ij}|x_{ij}|^2\ .
\end{align}
Here we used that $0=\Tr(X)=x_{11}+x_{22}$, and $X^*=X$, so $x_{12}=\overline{x_{21}}$.
On the other hand, let us denote a diagonal matrix
$$\Gamma^j=\sum_kC_{kj}\ket{\phi_k}\bra{\phi_k}\ . $$
Then $ D=\sum_j\ket{\psi_j}\bra{\psi_j}\Gamma^j ,$ and therefore by Cauchy-Schwatz inequality
\begin{align}
\left|\Tr\{D X)\} \right|^2&=\left|\sum_{ji}x_{ij}\bra{\psi_j}\Gamma^j\ket{\psi_i} \right|^2\\
&\leq \left(\sum_{ji}|x_{ij}|^2\right)\left(\sum_{ij}\left| \bra{\psi_i}\Gamma^j\ket{\psi_j} \right|^2\right)\\
&=\frac{1}{2}\|X\|_1^2 \sum_{ij} \bra{\psi_j}\Gamma^j\ket{\psi_i} \bra{\psi_i}\Gamma^j\ket{\psi_j} \\
&=\frac{1}{2}\|X\|_1^2 \sum_j \bra{\psi_j}(\Gamma^j)^2\ket{\psi_j}\\
&\leq \|X\|_1^2\, C^2\ .
\end{align}
The last inequality follows from the fact that $\Gamma^j\leq C I$.
And therefore,
$$\left|\Tr\{DX\} \right| \leq C \|X\|_1\ .$$

\end{proof}

In the most general case, unfortunately, we are picking up a factor of $\sqrt{d}$ in the upper bound. Note that the only instance where the conditions on $\rho$ and $\sigma$ were used in the proof of Theorem \ref{thm:main3} are in the proof of the Lemma \ref{lemma}. In the most general case, 
$$|\Tr(DX)|\leq \|DX\|_1\leq \|X\|_1\|D\|_\infty\leq \|X\|_1\|D\|_2\ . $$
And from the structure of $D$ in (\ref{eq:D1}),
$$\|D\|_2^2=\Tr(D^*D)=\sum_{kj}C_{kj}^2|\bra{\psi_j}\ket{\phi_k}|^2\leq C^2 d\ .$$
And therefore, using this result in the proof of Theorem \ref{thm:main3}, we obtain the following upper bound.

\begin{theorem}\label{thm:main}
Let $f\in\mathcal{Q}_{(0,\infty)}$ be an operator monotone decreasing function such that $f(1)=0$. Let $\rho$ and $\sigma$ be two strictly positive density operators on a $d$-dimensional Hilbert space. Assume one of the two conditions: 1) $\rho$ is full rank; 2) $f$ is such that $a_f=0$ (defined below). Then 
\begin{equation}\label{eq:main}
S_f(\rho\|\sigma)\leq\|\rho-\sigma\|_1\sqrt{d}\left[\frac{\lambda_\rho}{\lambda_\rho-\alpha_\sigma} f(\lambda_\rho^{-1}\alpha_\sigma)-a_f\right] \ ,
\end{equation}
where

\begin{itemize}
\item $\lambda_\rho\in(0,1]$ is the largest eigenvalue of $\rho$,
\item $\alpha_\sigma\in(0,1]$ is the smallest eigenvalue of $\sigma$,
\item $a_f=-\lim_{y\uparrow \infty}\frac{f(iy)}{iy}.$
\end{itemize}
\end{theorem}

We conjecture that the dimensionless bound holds in any dimension without any restriction on the states.
\begin{conjecture}\label{conjecture}
Let $\rho=\sum_j\lambda_j\ket{\psi_j}\bra{\psi_j},$ and $\sigma=\sum_k\mu_k\ket{\phi_k}\bra{\phi_k}$ be written in their spectral decomposition. Let
\begin{equation}\label{eq:D}
D=\sum_{kj}C_{kj}\bra{\psi_j}\ket{\phi_k}\ket{\psi_j}\bra{\phi_k}\ ,
\end{equation}
such that $0\leq C_{kj}\leq C$ for all $k,j$ and some $C$.  Then
$$|\Tr(D(\rho-\sigma))|\leq C\|\rho-\sigma\|_1\ . $$
\end{conjecture}
Note that with the above notations,
$$\Tr(D(\rho-\sigma))=\sum_{kj}C_{kj}(\lambda_j-\mu_k)|\bra{\phi_k}\ket{\psi_j}|^2 \ .$$

\section{Relative entropy}\label{sec:rel-better}
The proof of the following upper bound on the relative entropy is inspired by the proof in \cite{R10}.

\begin{theorem}\label{thm:main2}
Let $\rho$ and $\sigma$ be two strictly positive density operators on a finite-dimensional Hilbert space. Then 
\begin{equation}\label{eq:main2}
S(\rho\|\sigma)\leq\|\rho-\sigma\|_1\lambda_\rho \frac{\log(\alpha_\rho)-\log(\alpha_\sigma)}{\alpha_\rho-\alpha_\sigma}\leq\|\rho-\sigma\|_1\frac{\lambda_\rho}{\alpha}  \ ,
\end{equation}
where

\begin{itemize}
\item $\lambda_\rho\in(0,1]$ is the largest eigenvalue of $\rho$,
\item $\alpha=\min\{\alpha_\rho, \alpha_\sigma\}\in(0,1]$ is the minimum between the smallest  eigenvalues of $\rho$ and $\sigma$.
\end{itemize}
\end{theorem}

\begin{proof}{\it(of Theorem \ref{thm:main2})} 
The relative entropy is
\begin{align}
S(\rho\|\sigma)&=\Tr\{(\log\rho-\log\sigma)\rho \}\ .
\end{align}

By (\ref{eq:logex}) logarithm admits the following representation:
\begin{equation}\label{eq:g2}
-\log(x)=\int_0^\infty\left(\frac{1}{t+x}-\frac{t}{t^2+1} \right)dt\ . 
\end{equation}
Therefore,
\begin{equation}\label{eq:S2}
S(\rho\|\sigma)=\int_0^\infty\Tr\left\{\left(\frac{1}{t+\sigma}-\frac{1}{t+\rho} \right)\rho\right\} dt\ .
\end{equation}

Note that the formula $A^{-1}-B^{-1}=A^{-1}(B-A)B^{-1}$ holds for any invertible operators $A$ and $B$.  Therefore, 
\begin{align}
&\left| \int_0^\infty\Tr\left\{\left(\frac{1}{t+\rho}-\frac{1}{t+\sigma} \right)\rho\right\} dt\right|\leq \int_0^\infty\Tr\left|({t+\rho})^{-1}(\sigma-\rho)({t+\sigma})^{-1}\rho\right| dt\ .
\end{align}
Since for two operators $\|XY\|_\infty\leq \|X\|_\infty\|Y\|_\infty$, we have that the last line can be bounded as
\begin{align}
&\leq \int_0^\infty \|\rho-\sigma\|_1 \|(t+\rho)^{-1}\|_\infty\|({t+\sigma})^{-1}\|_\infty\|\rho\|_\infty dt\\
&= \|\rho-\sigma\|_1\lambda_\rho \int_0^\infty \frac{1}{t+\alpha_\sigma}\frac{1}{t+\alpha_\rho}dt\\
&= \|\rho-\sigma\|_1\frac{1}{\alpha_\rho-\alpha_\sigma}\lambda_\rho \int_0^\infty\left( \frac{1}{t+\alpha_\sigma}-\frac{1}{t+\alpha_\rho}\right)dt\ .
\end{align}

From (\ref{eq:g2}) we have
\begin{align}
&\int_0^\infty\left( \frac{1}{t+\alpha_\sigma}-\frac{1}{t+\alpha_\rho}\right)dt=\log(\alpha_\rho)-\log(\alpha_\sigma)\ ,
\end{align}
and therefore 
$$S_f(\rho\|\sigma)\leq  \|\rho-\sigma\|_1\lambda_\rho \frac{\log(\alpha_\rho)-\log(\alpha_\sigma)}{\alpha_\rho-\alpha_\sigma}\ .$$

Without loss of generality assume that $\alpha_\rho\leq \alpha_\sigma$. By the Mean Value Theorem, there exists $c\in[\alpha_\rho,\alpha_\sigma]$ such that
$$\frac{\log(\alpha_\rho)-\log(\alpha_\sigma)}{\alpha_\rho-\alpha_\sigma}=\frac{1}{c}\leq\frac{1}{\alpha_\rho}\ .$$
This leads to the statement in the Theorem.
\end{proof}

\begin{remark}\label{rem:rel-better}
Let us give an example when the derived bound (\ref{eq:main2}) is better than the known bound for the relative entropy (\ref{eq:rel_known}). Let $\cH$ be a $d$-dimensional Hilbert space with $d\geq 5$, with the orthonormal basis denoted as $\{\ket{\psi_j}\}_{j=1}^d$. Let us take the following states
$$\rho=\sum_{j=1}^d \frac{1}{d}\ket{\psi_j}\bra{\psi_j}\ , $$
and 
$$\sigma=\frac{1}{d}\ket{\psi_1}\bra{\psi_1}+\left(1-\frac{1}{d}\right)\ket{\psi_2}\bra{\psi_2}\ . $$
Then the trace distance between these states is
$$\|\rho-\sigma\|_1=\left(1-\frac{2}{d}\right)+(d-2)\frac{1}{d}=2-\frac{4}{d}\ . $$
The upper bound in (\ref{eq:rel_known}) is 
$$S(\rho\|\sigma)\leq\frac{1}{2}\|\rho-\sigma\|_1\log(1+d\|\rho-\sigma\|_1/2)=\frac{1}{2}\|\rho-\sigma\|_1\log(d-1)\ .$$
The upper bound that we derived in  (\ref{eq:main2}) gives
$$S(\rho\|\sigma)\leq\|\rho-\sigma\|_1\ . $$
For $d\geq 5$, 
$$\frac{1}{2}\log(d-1)\geq 1\ . $$
This shows that for these states, our new bound is better than the old one. Clearly, these are not the only states for which the new bound is better, but these states give an example when it is.
\end{remark}

\begin{corollary}\label{cor:rel}
If $\rho$ and $\sigma$ are qubits, then
$$ S(\rho\|\sigma)\leq  \|\rho-\sigma\|_1 \lambda_\rho \frac{\log\lambda_\rho-\log\alpha_\sigma}{\lambda_\rho-\alpha_\sigma}\leq \|\rho-\sigma\|_1 \frac{\lambda_\rho}{\alpha_\sigma}\ . $$
\end{corollary}
\begin{proof}
If $\rho$ and $\sigma$ are two-dimensional, then from Theorem \ref{thm:main3}, we obtain
$$S(\rho\|\sigma)\leq \|\rho-\sigma\|_1 \lambda_\rho \frac{\log\lambda_\rho-\log\alpha_\sigma}{\lambda_\rho-\alpha_\sigma}\ . $$
Since $\Tr\rho=\Tr\sigma=1$, we have  $\alpha_\sigma\leq \lambda_\rho$. By the Mean Value Theorem for the logarithmic function, there exists $c\in[\alpha_\sigma, \lambda_\rho]$ such that
$$ \frac{\log\lambda_\rho-\log\alpha_\sigma}{\lambda_\rho-\alpha_\sigma}=c^{-1}\leq \alpha_\sigma^{-1}\ . $$
Therefore,
$$ S(\rho\|\sigma)\leq \|\rho-\sigma\|_1 \frac{\lambda_\rho}{\alpha_\sigma}\ . $$
\end{proof}

\section{ Tsallis entropy for $q>1$}
\begin{remark} \label{rem}
{Let us look at the inequality (3.7) in \cite{R10},} which is used in the derivation of the bound (\ref{eq:Rmoreold}): the inequality states
\begin{equation}\label{eq:rem1}
\frac{\alpha_\rho^{-r}-\alpha_\sigma^{-r}}{\alpha_\sigma-\alpha_\rho}\leq \frac{1}{\alpha^q}\ ,
\end{equation}
where $r:=q-1\in(0,1]$ for $q\in(1,2].$

The function $f(x):=-x^{-r}$ is concave and monotonically increasing for $r\in(0,1].$ Then by the Mean Value Theorem, there exists a point $c$ between points $\alpha_\rho$ and $\alpha_\sigma$, such that
$$\frac{f(\alpha_\sigma)-f(\alpha_\rho)}{\alpha_\sigma-\alpha_\rho}=f'(c)\leq f'(\alpha)=r\alpha^{-r-1}=(q-1)\alpha^{-q}\ , $$
where $\alpha=\min\{ \alpha_\rho, \alpha_\sigma\}$. This improves the constant in (\ref{eq:rem1}), leading to the bound (\ref{eq:Rmore}).
\end{remark}

\section{Tsallis entropy for $q\in(0,1)$}

For $q\in(0,1)$ the function $f(x)=\frac{1}{1-q}(1-x^{1-q})$ is an operator monotone decreasing function, which defines the quasi-entropy
$$S_q(\rho\|\sigma)=\frac{1}{1-q}(1-\Tr(\rho^q\sigma^{1-q}))\ . $$

\begin{theorem}\label{thm:main1}
Let $\rho$ and $\sigma$ be two strictly positive density operators on a finite-dimensional Hilbert space. Then 
\begin{equation}\label{eq:main1}
S_q(\rho\|\sigma)\leq \frac{1}{1-q} \|\rho-\sigma\|_1\lambda_\rho^q \frac{\alpha_\rho^{1-q}-\alpha_\sigma^{1-q}}{\alpha_\rho-\alpha_\sigma}\leq\|\rho-\sigma\|_1\frac{\lambda_\rho^q}{\alpha^q}\ ,
\end{equation}
where

\begin{itemize}
\item $\lambda_\rho$ is the largest eigenvalue of $\rho$,
\item $\alpha=\min\{\alpha_\rho,\alpha_\sigma\} \in(0,1]$ is the minimum between the smallest eigenvalues of $\rho$ and $\sigma$.
\end{itemize}
\end{theorem}

\begin{proof}
With the above function $f$, let us denote $g(x)=\frac{1}{1-q}x^{1-q}$, which is an operator monotone function. Since $S_f(\rho\|\rho)=0= \Tr(f(\Delta_{\rho,\rho})\rho)$, we have
\begin{align}
S_q(\rho\|\sigma)=S_f(\rho\|\sigma)&=\Tr\{f(\Delta_{\sigma,\rho})\rho)- \Tr(f(\Delta_{\rho,\rho})\rho\}\\
&=\Tr\{(f(L_\sigma R_{\rho^{-1}})-f(L_\rho R_{\rho^{-1}}))\rho \}\\
&=\Tr\{(g(\rho )-g(\sigma))\rho^q \}\ .
\end{align}
Here we used  that  the modular operator is the product of left and right multiplications, $\Delta_{\sigma,\rho}=L_\sigma R_{\rho^{-1}}$.

Since  $g(x)$ be an operator monotone function, from (\ref{low2}), it admits the following integral representation (see also (\ref{eq:powex})) :
\begin{equation}\label{eq:g}
g(x)=g(1)+a_g(x-1)-\int_0^\infty\left(\frac{1}{t+x}-\frac{t}{t^2+1} \right)d\mu_g(t)\ , 
\end{equation}
where $a_g\geq0$.
Therefore,
\begin{equation}\label{eq:S}
S_f(\rho\|\sigma)=a_g\Tr\left\{(\sigma-\rho) h({\rho^{-1}})\rho\right\} -\int_0^\infty\Tr\left\{\left(\frac{1}{t+\rho}-\frac{1}{t+\sigma} \right)\rho^q\right\} d\mu_g(t)\ .
\end{equation}

By \cite{W18}, for any operators $X, Y, Z$ the following bound holds 
\begin{equation}\label{eq:three}
|\Tr(XYZ)|\leq \|X\|_\infty\|Z\|_\infty \Tr|Y|\ . 
\end{equation}

For the first term in (\ref{eq:S}), we use the bound above for two operators:
\begin{equation}\label{eq:first}
\left| \Tr\left\{(\sigma-\rho)\rho^q\right\} \right|\leq \|\rho-\sigma\|_1\|\rho^q\|_\infty=  \|\rho-\sigma\|_1\lambda_\rho^q\ .
\end{equation}

For the second term, we note that the formula $A^{-1}-B^{-1}=A^{-1}(B-A)B^{-1}$ holds for any invertible operators $A$ and $B$.  Therefore, 
\begin{align}
&\left| \int_0^\infty\Tr\left\{\left(\frac{1}{t+\rho}-\frac{1}{t+\sigma} \right)\rho^q\right\} d\mu_g(t)\right|\leq \int_0^\infty\Tr\left|({t+\rho})^{-1}(\sigma-\rho)({t+\sigma})^{-1}\rho^q\right| d\mu_g(t)\ .
\end{align}
Applying (\ref{eq:three}) and the fact that for two operators $\|XY\|_\infty\leq \|X\|_\infty\|Y\|_\infty$, we have that the last line can be bounded as
\begin{align}
&\leq \int_0^\infty \|\rho-\sigma\|_1 \|(t+\rho)^{-1}\|_\infty\|({t+\sigma})^{-1}\|_\infty\| \rho^q\|_\infty d\mu_g(t)\\
&= \|\rho-\sigma\|_1\lambda_\rho^q\int_0^\infty \frac{1}{t+\alpha_\sigma}\frac{1}{t+\alpha_\rho}d\mu_g(t)\\
&= \|\rho-\sigma\|_1\frac{1}{\alpha_\rho-\alpha_\sigma}\lambda_\rho^q \int_0^\infty\left( \frac{1}{t+\alpha_\sigma}-\frac{1}{t+\alpha_\rho}\right)d\mu_g(t)\ .
\end{align}

Note that from (\ref{eq:g}) we have
\begin{align}
&\int_0^\infty\left( \frac{1}{t+\alpha_\sigma}-\frac{1}{t+\alpha_\rho}\right)d\mu_g(t)=g(\alpha_\rho)-g(\alpha_\sigma)-a_g(\alpha_\rho-\alpha_\sigma)\ ,
\end{align}
and therefore the second term can be bounded by
\begin{equation}\label{eq:second}
\|\rho-\sigma\|_1\lambda_\rho^q \frac{g(\alpha_\rho)-g(\alpha_\sigma)}{\alpha_\rho-\alpha_\sigma}-a_g\|\rho-\sigma\|_1\lambda_\rho^q\ .
\end{equation}

Putting (\ref{eq:first}) and (\ref{eq:second}) together, we find
$$S_f(\rho\|\sigma)\leq  \|\rho-\sigma\|_1\lambda_\rho^q \frac{g(\alpha_\rho)-g(\alpha_\sigma)}{\alpha_\rho-\alpha_\sigma}\ .$$

Without loss of generality assume that $\alpha_\rho\leq \alpha_\sigma$. By the Mean Value Theorem, there exists $c\in[\alpha_\rho,\alpha_\sigma]$ such that
$$\frac{g(\alpha_\rho)-g(\alpha_\sigma)}{\alpha_\rho-\alpha_\sigma}=g'(c)\ .$$
Since $g$ is concave, the derivative $g'$ is monotonically decreasing, therefore
$$g'(c)\leq g'(\alpha_\rho)=\alpha_\rho^{-q}\ . $$
Thus we arrive at the statement in the Theorem.
\end{proof}

\section{Tsallis entropy for qubits}\label{sec:Ts-qubits}
For $q\in(0,2)$, function $f(x)=\frac{1}{1-q}(1-x^{1-q})$ is operator monotone decreasing.
If $\rho$ and $\sigma$ are two-dimensional states, then from Theorem \ref{thm:main}, we obtain
\begin{align}
S_{q}(\rho\|\sigma)&\leq \frac{1}{1-q}\, \|\rho-\sigma\|_1\, \frac{\lambda_\rho}{\lambda_\rho-\alpha_\sigma}(1-\lambda_\rho^{q-1}\alpha_\sigma^{1-q})\\
&=\frac{1}{1-q}\,\lambda_\rho^{q} \|\rho-\sigma\|_1\, \frac{\lambda_\rho^{1-q}-\alpha_\sigma^{1-q}}{\lambda_\rho-\alpha_\sigma}\ .
\end{align}
Since $\Tr\rho=\Tr\sigma=1$, we have $\alpha_\sigma\leq\lambda_\rho$. By the Mean Value Theorem, there exists $c$ between points $\alpha_\sigma$ and $\lambda_\rho$ such that
$$\frac{\lambda_\rho^{1-q}-\alpha_\sigma^{1-q}}{\lambda_\rho-\alpha_\sigma}=(1-q)c^{-q}\ ,$$
and $$c^{-q}\leq \alpha_\sigma^{-q}\ .  $$
Therefore,
$$ S_{q}(\rho\|\sigma)\leq  \|\rho-\sigma\|_1\,\frac{\lambda_\rho^{q}}{\alpha_\sigma^{q}} \ . $$

\vspace{0.3in}
\textbf{Acknowledgments.}  A. V. is partially supported by NSF grant DMS-1812734.

\end{document}